\documentclass{revtex4}
\usepackage{amsmath,amsthm}
\usepackage{setspace}
\usepackage{graphicx}
\newcommand{\bra}[1]{\langle#1|}
\newcommand{\ket}[1]{|#1\rangle}
\usepackage{amsfonts}
\newtheorem{theorem}{Theorem}[section]

\newtheorem{proposition}[theorem]{Proposition}

\theoremstyle{remark}
\newtheorem{remark}[theorem]{Remark}

\theoremstyle{definition}

\theoremstyle{example}
\newtheorem{example}[theorem]{Example}

\theoremstyle{notation}

\begin{document}

\title{Analytic representations with Theta functions for systems on ${\mathbb Z}(d)$ and on ${\mathbb S}$ }
\author{P. Evangelides, C. Lei, A. Vourdas\\
Department of Computing\\University of Bradford\\ Bradford BD7 1DP, UK}
\begin{abstract}
An analytic representation with Theta functions on a torus, for systems with variables in ${\mathbb Z}(d)$, is considered.
Another analytic representation with Theta functions on a strip, for systems with positions in a circle ${\mathbb S}$ and momenta in ${\mathbb Z}$, is also considered.
The reproducing kernel formalism for these two systems is studied.
Wigner and Weyl functions in this language, are also studied.
\end {abstract}
\maketitle
\section{Introduction}
Analytic representations in quantum mechanics, represent the quantum states with analytic functions.
Then the powerful theory of analytic functions can be used in a quantum mechanical context.
Various analytic representations have been studied in the literature: the Bargmann representation in the complex plane
for the harmonic oscillator\cite{B}, analytic representations in the unit disc for systems with $SU(1,1)$ symmetry,
analytic representations in the extended complex plane for systems with $SU(2)$ symmetry, etc
(for reviews see \cite{P,B1,H,V}).

Refs.\cite{ZV,TVZ,CGV} have used an analytic representation based on theta functions\cite{T1,T2,T3} on a torus, for systems with 
variables in ${\mathbb Z}(d)$ (the ring of integers modulo $d$).
The $d$ zeros of the analytic function were used to describe the time evolution of these systems in terms of $d$ paths in the torus.
A related representation has also been used in \cite{L} in studies of chaos.
Work on other aspects of systems with finite Hilbert space have been reviewed in \cite{R1,R2,R3,R4,R5,R6,R7}.
It is known that there are differences in the formalism in the cases that $d$ is an odd or even number, and here
we consider the case of odd $d$. In this case the inverse of $2$ in ${\mathbb Z}(d)$ exists (if $d=2j+1$ then $2^{-1}=j+1$), and it 
enters in many of the formulas below.

In this paper we study this analytic representation on a torus, as a subject in its own right analogous to the Bargmann formalism.
In particular we develop a reproducing kernel formalism in this context.
We also introduce an analogous formalism on a strip for quantum systems on a circle.
Quantum mechanics on a circle, has been studied for a long time \cite{c1,c2,c3,c4,c5,c6,c7,c8}, and
coherent states on a circle have been considered in \cite{co1,co2,co3,co4}.
Our approach complements this work, using an analytic language with Theta functions.
For simplicity, we consider periodic boundary conditions (zero Aharonov-Bohm magnetic flux).

In section II, we introduce briefly the basic formalism for both finite systems with variables in ${\mathbb Z}(n)$, and 
also systems with positions on  circle ${\mathbb S}$ and momenta in ${\mathbb Z}$,
in order to define the notation. 
In section III we introduce an analytic representation on a torus for systems with variables in ${\mathbb Z}(n)$, using Theta functions.
We also introduce an analytic representation on a strip for systems on a circle, using Theta functions.
In section IV we study the reproducing kernel formalism for finite systems.
Analogous results for systems on a circle, are presented in section V.
Proposition \ref {pro1} for finite systems and \ref{pro10} for the circle, are the main results of this paper.   
In section VI we study Wigner and Weyl functions in this language.
We conclude in section VII with a discussion of our results.

\section{Preliminaries}

\subsection{Finite quantum systems}

We consider a finite quantum system with variables in ${\mathbb Z}(d)$, described with 
Hilbert space ${\cal H}(d)$ of dimension $d$, which we assume to be an odd integer.
Let $\ket{X; m}$ where $m\in {\mathbb Z}(d)$, be a basis  which we call position states.
With a finite Fourier transform we get the momentum basis
\begin{eqnarray}\label{Fou}
\ket{P;n}={\cal F}|{X};n\rangle ;\;\;\;\;\;\;{\cal F}=d^{-1/2}\sum _{m,n}\omega (mn)\ket{X;m}\bra{X;n}
;\;\;\;\;\;\omega (m)=\exp \left [i \frac{2\pi m}{d}\right ]
\end{eqnarray}

The position-momentum phase space in this case is the toroidal
lattice ${\mathbb Z}(d)\times {\mathbb Z}(d)$. Displacement operators in this phase space are defined as
\begin{eqnarray}\label{99}
&&{\cal D}(\alpha,\beta)={\cal Z}^\alpha {\cal X}^\beta \omega (-2^{-1}\alpha \beta );\;\;\;\;\;\alpha, \beta \in {\mathbb Z}(d)\nonumber\\
&&{\cal Z}=\sum _{n}\omega (n)|{X};n\rangle \langle {X};n|=\sum _n\ket{{P};n+1}\bra{{P};n}\nonumber\\
&&{\cal X}=\sum _{n}\omega (-n)|{P};n\rangle \langle {P};n|=\sum _n \ket{{X};n+1}\bra{{ X};n}
\end{eqnarray}
They obey the relations
\begin{eqnarray}\label{58}
{\cal X}^{d}&=&{\cal Z}^{d}={\bf 1};\;\;\;\;\;\;
{\cal X}^\beta {\cal Z}^\alpha = {\cal Z}^\alpha {\cal X}^\beta \omega (-\alpha \beta);\;\;\;\;\;\;
\end{eqnarray}
We also define the displaced Fourier operator and the displaced parity operator:
\begin{eqnarray}\label{4}
{\cal F}(\alpha,\beta)&=&{\cal D}(\alpha,\beta){\cal F}{\cal D}(-\alpha,-\beta)
=\omega [2^{-1}(\alpha ^2+\beta ^2)]{\cal F}{\cal D}(-\alpha -\beta,\alpha -\beta)\nonumber\\&=&
\omega [2^{-1}(\alpha ^2+\beta ^2)]{\cal D}(\alpha -\beta,\alpha +\beta){\cal F}
\nonumber\\
{\cal P}(\alpha,\beta)&=&{\cal D}(\alpha,\beta){\cal F}^2{\cal D}(-\alpha,-\beta)={\cal D}(2\alpha,2\beta){\cal P}(0,0)
\end{eqnarray}
and we can show that
\begin{eqnarray}\label{45}
{\cal P}(\gamma, \delta)=\frac{1}{d} \sum _{\alpha, \beta}\omega (\beta \gamma -\alpha \delta){\cal D}(\alpha,\beta)
\end{eqnarray}

\subsection{Quantum systems on a circle}
A particle on a circle ${\mathbb S}$ with radius $1$, is described with the wavefunction $q(x)$ where
\begin{equation}\label{Bloch}
|q\rangle=\frac{1}{2\pi}\int_0^{2\pi}dx\;q(x)|x\rangle;\;\;\; q(x)=\bra{x} q\rangle;\;\;\;\;
q(x+2\pi )=q(x);\;\;\;\;\frac {1}{2\pi }\int_0^{2\pi } |q(x)|^2 dx =1,
\end{equation}
With a Fourier expansion, we get
\begin{equation}\label{Rx}
q(x)=\sum_{N=-\infty}^{\infty} q_N \exp (i N x);\;\;\;q_N=\frac{1}{2\pi}\int _0^{2\pi}q(x) \exp (-iNx)dx.
\end{equation}
Let $|x \rangle$, $|N \rangle$ be position and momentum eigenstates. Then:
  \begin{eqnarray}
   |N \rangle &=&\frac{1}{2\pi }\int _0 ^{2\pi }dx\exp (iN x)|x \rangle \nonumber\\
\bra{x}y\rangle&=&2\pi\delta (x-y)\nonumber\\
\frac{1}{2\pi}\int_0^{2\pi}\ket{x}\bra{x}dx&=&\sum _{N=-\infty}^{\infty}\ket {N}\bra{N}={\bf 1}
  \end{eqnarray}
where $\delta (x-y)$ is the `comb delta function' with period $2\pi$ ($x,y \in {\mathbb R}/2\pi {\mathbb Z}$). It satisfies the relation
 \begin{eqnarray}
  \sum _{N=-\infty}^{\infty}\exp [iN(x-y)]=2\pi \delta(x-y)
   \end{eqnarray}
The scalar product is given by
 \begin{eqnarray}
\bra{q_2}q_1\rangle=\frac{1}{2\pi}\int _0^{2\pi}[q_2(x)]^*q_1(x)dx
   \end{eqnarray}

Displacement operators in the ${\mathbb S}\times {\mathbb Z}$ phase space are given by\cite{c8}
  \begin{eqnarray}\label{pa2}
   D(a, K)\ket{x}&=& \exp \left [{iK} \left (x+\frac{a}{2} \right ) \right ]|x+a \rangle;\;\;\;\;\;
K\in{\mathbb Z}\nonumber\\
   D(a, K)\ket{N}&=&\exp\left [-ia \left (N+\frac{K}{2}\right )\right ]|{N+K} \rangle\nonumber\\
   D(a, K)D(b, M)&=&D(a+b, K+M) \exp \left[\frac{i}{2}( Kb -Ma) \right]\nonumber\\
   D^\dagger(a,K)&=&D(-a,-K)
  \end{eqnarray}
$D(a,K)$ is periodic in $a$, with period
is $2\pi $ if $K$ is even and $4\pi $ if $K$ is odd:
 \begin{eqnarray}\label{pa100}
   D(a +2\pi, K)&=& (-1)^KD(a , K)
  \end{eqnarray}

 We also define the parity operator as:
 \begin{eqnarray}
    U_0 = \frac{1}{2\pi } \int^{2\pi }_{0} |x\rangle \langle -x| 
          dx = \sum_{N=-\infty}^{\infty} |{-N}\rangle \langle N|=\frac{1}{2\pi}\sum_{K=-\infty}^\infty \int_0^{2\pi}da\;D(a,2K).  \label{U0}
   \end{eqnarray}
We note that only the $D(a,2K)$, with even displacements $2K$ in the momentum direction, appear in the right hand side.

We consider the displaced parity operator 
\begin{eqnarray}\label{pioiu}
{U}(a,K)&=&D(a,K){U}_0=U_0D(-a,-K)\nonumber\\
{U}(a+2\pi,K)&=&(-1)^K{U}(a,K);\;\;\;\;[U(a,K)]^2={\bf 1}
\end{eqnarray}
For even and odd $K$ we get
\begin{eqnarray}
&&{U}(a,2K)=D\left (\frac{a}{2},K\right )U_0D\left (-\frac{a}{2},-K\right )\nonumber\\
&&{U}(a,2K+1)=D\left (\frac{a}{2},K\right )U_0D\left (-\frac{a}{2},-K-1\right )\exp\left (\frac{ia}{4}\right ).
\end{eqnarray}
There is an asymmetry in the even and odd cases in these formulas.

It is related to the displacement operator, through the Fourier transform:
\begin{eqnarray}\label{fou}
{U}(a,K)=\frac{1}{2\pi}\sum_{M=-\infty}^\infty\int_0^{2\pi}db\;D(b,K+2M)\exp\left [\frac{i}{2}(Kb-aK-2Ma)\right ].
\end{eqnarray}

\section{Analytic representations}
\subsection{Analytic representations on a torus for finite quantum systems }

Let $\ket {g}$ be an arbitrary pure normalized state
\begin{eqnarray}\label{6}
&&\ket {g}=\sum _m g_m\ket {X;m}=\sum _m{\widetilde g}_m\ket {P;m};\;\;\;\;\;\;\sum _m |g_m|^2=1\nonumber\\
&&{\widetilde g}_m=d^{-1/2}\sum _n\omega (-mn)g_n
\end{eqnarray}
We use the notation
\begin{eqnarray}
\ket {g^*}=\sum _m g_m^*\ket {X;m};\;\;\;\;\;\;
\bra {g}=\sum _m g_m^*\bra {X;m};\;\;\;\;\;\;\bra {g^*}=\sum _m g_m\bra {X;m}
\end{eqnarray}
In ref\cite{ZV} we represented the state $\ket{g}$ of Eq.(\ref{6}), with the function
\begin{eqnarray}\label{aaa1}
G(z)=\pi^{-1/4} \sum_{m=0}^{d-1} g_m\;\Theta_3 \left [\frac{\pi m}{d}-z\sqrt{\frac{\pi}{2d}};\frac{i}{d}\right ]
\end{eqnarray}
where $\Theta_3$ is Theta function defined as
\begin{eqnarray}\label{pa4}
\Theta _3(u,\tau)=\sum_{n=-\infty}^{\infty}\exp(i\pi \tau n^2+i2nu).
\end{eqnarray}
It is known that
\begin{eqnarray}\label{100}
\Theta _3(u,\tau)=(-i\tau)^{-1/2}\exp \left(\frac{u^2}{i\pi \tau}\right)\Theta _3\left (\frac{u}{\tau},-\frac{1}{\tau}\right)
\end{eqnarray}
$G(z)$ is an analytic function. 
The scalar product is given by 
\begin{eqnarray}\label{scalar}
 \langle g_1^\ast| g_2 \rangle &=& \frac{1}{d^{3/2}\sqrt{2\pi} }\int_S  d\mu (z)G_1(z) G_2(z^\ast)
\end{eqnarray}
$z_R$ and $z_I$ are the real and imaginary parts of $z$, correspondingly.
These relations are proved using
the orthogonality relation\cite{TVZ}
\begin{eqnarray}
 2^{-1/2}\pi ^{-1}d^{-3/2}\int_S d\mu (z)
\Theta_3 \left [\frac{\pi n}{d}-z\sqrt{\frac{\pi}{2d}};\frac{i}{d}\right ]\;
\Theta_3 \left [\frac{\pi m}{d}-z^*\sqrt{\frac{\pi}{2d}};\frac{i}{d}\right ]
=\delta (m,n)
\end{eqnarray}

Using the properties of Theta functions we prove that
\begin{eqnarray} \label{periodicity}
 &&G( z+\sqrt{2\pi d} ) = G(z)\nonumber\\
 &&G( z+ i \sqrt{2\pi d})  = G(z)\exp \left (\pi d-i z\sqrt{2\pi d} \right ).
\end{eqnarray}
It is seen that $G(z)$ is defined on a cell $S=[M\sqrt{2\pi d},(M+1)\sqrt{2\pi d})\times [N\sqrt{2\pi d},(N+1)\sqrt{2\pi d})$ where $(M,N)$ are integers labelling the cell.
We call ${\cal A}$ the space of these functions.

The coefficients $g_m$, ${\widetilde g}_m$ in Eq.(\ref{6}) are given by 
\begin{eqnarray}\label{5b}
g_m&=& 2^{-1/2}\pi ^{-3/4}d^{-3/2}\int_S d\mu (z)
\Theta_3 \left [\frac{\pi m}{d}-z\sqrt{\frac{\pi}{2d}};\frac{i}{d}\right ]\;G(z^*);\;\;\;\;\;d\mu (z)=d^2z \exp \left( - z_I^2 \right)\nonumber\\
{\widetilde g}_m&=& 2^{-1/2}\pi ^{-3/4}d^{-3/2}d^{-1/2}\sum _n\omega (-mn)\int_S d\mu (z)
\Theta_3 \left [\frac{\pi n}{d}-z\sqrt{\frac{\pi}{2d}};\frac{i}{d}\right ]\;G(z^*).
\end{eqnarray}
\begin{example}
The momentum states $\ket{P;k}$ are represented with the function
\begin{eqnarray}\label{aaa1}
G(z;k)=\pi ^{-1/4}\frac{1}{\sqrt{d}}\sum _m\omega (km)\Theta_3 \left [\frac{\pi m}{d}-z\sqrt{\frac{\pi}{2d}};\frac{i}{d}\right ]=\pi ^{-1/4}
\exp \left ( -\frac{z^2}{2}\right )\Theta_3 \left [\frac{\pi k}{d}-iz\sqrt{\frac{\pi}{2d}};\frac{i}{d}\right ]
\end{eqnarray}
In order to prove this we use the
\begin{eqnarray}
&&\frac{1}{\sqrt{d}}\sum _m\omega (km)\Theta_3 \left [\frac{\pi m}{d}-z\sqrt{\frac{\pi}{2d}};\frac{i}{d}\right ]\nonumber\\
&&=\frac{1}{\sqrt{d}}\sum _{n=-\infty}^{\infty}\exp \left (-\frac{\pi n^2}{d}-izn\sqrt{\frac{2\pi}{d}}\right)\sum _{m=0}^{d-1}\omega [m(n+k)]\nonumber\\
&&=\sqrt{d}\sum _{n=-k+dN}\exp \left (-\frac{\pi n^2}{d}-izn\sqrt{\frac{2\pi}{d}}\right)\nonumber\\
&&=\sqrt{d}\exp\left (-\frac{\pi k^2}{d}+ikz\sqrt{\frac{2\pi}{d}}\right )\Theta_3 \left [-i\pi k-z\sqrt{\frac{\pi d}{2}};{i}{d}\right ] 
\end{eqnarray}
Using Eq.(\ref{100}) we prove Eq.(\ref{aaa1}).
\end{example}
\begin{proposition}
\begin{itemize}
\mbox{}
\item[(1)]
The analytic function $G(z)$ has exactly $d$ zeros $\zeta _{\nu}$ in each cell $S=[M\sqrt{2\pi d},(M+1)\sqrt{2\pi d})\times [N\sqrt{2\pi d},(N+1)\sqrt{2\pi d})$
which obey the constraint
\begin{eqnarray}\label{con}
\sum _{\nu =1}^d \zeta _{\nu}=\sqrt{2\pi d}(M+iN)+d^{3/2}\sqrt{\frac{\pi }{2}}(1+i)
\end{eqnarray} 
\item[(2)]
If the $d$ zeros $\zeta _{\nu}$ are given (and obey the above constraint),
then the function $G(z)$ is given by
\begin{eqnarray}\label{500}
G(z) &=& C\exp \left[ -i\sqrt{\frac{2\pi}{d}} Nz\right]
\prod_{n=1}^{d}  \Theta_3 \left[w_n(z);\; i\right]\nonumber\\
w_n(z)&=&\sqrt{\frac{\pi}{2d}}(z-\zeta _n)+\frac{\pi (1+i)}{2}
\end{eqnarray}
Here $N$ is the integer in the constraint of Eq.(\ref{con}), 
and $C$ is a constant determined by the normalization condition.
\end{itemize}
\end{proposition}
\begin{proof}
The proof has been given in \cite{ZV}. Eq.(\ref{con}) has also been given in \cite{L}.
\end{proof}
We note that in finite systems the $d-1$ zeros define uniquely the state (the last zero is determined from Eq.(\ref{con})).
In infinite systems the zeros do not define uniquely the state.

\subsection{Analytic representations on a strip for systems on a circle}

The state $\ket{q}$, is represented with the function $q(x)$ in the $x$-representation, and it is now represented with the analytic function
\begin{eqnarray}\label{A1}
Q(z)=\int_{0}^{2\pi}dx q(x)\Theta_3\left [\frac{x-z}{2};\frac{i}{2\pi} \right ].
\end{eqnarray}
The integrand in this integral is periodic with period $2\pi$.
The function $Q(z)$ is periodic:
\begin{eqnarray}\label{A2}
Q(z+2\pi)=Q(z).
\end{eqnarray}
Therefore it is sufficient to define the function $Q(z)$ on the strip $A=[0,2\pi]\times {\mathbb R}$ in the complex plane.

As examples we consider the states
\begin{eqnarray}
&&\ket{x}\;\;\rightarrow\;\;2\pi\Theta_3\left [\frac{x-z}{2};\frac{i}{2\pi} \right ]\nonumber\\
&&\ket{N}\;\;\rightarrow\;\;2\pi \exp\left (-\frac{N^2}{2}+iNz\right)
\end{eqnarray}
\begin{proposition}[orthogonality relation]
\begin{eqnarray}
\int _Adm(z)\Theta_3\left [\frac{x-z}{2};\frac{i}{2\pi} \right ]\Theta_3\left [\frac{y-z^*}{2};\frac{i}{2\pi} 
\right ]=\delta (x-y);\;\;\;\;dm(z)=\frac{1}{4\pi ^{5/2}}\exp (-z_I^2)d^2z
\end{eqnarray}
\end{proposition}
\begin{proof}
Using  the definition of Theta functions in Eq.(\ref{pa4}), we get
\begin{eqnarray}\label{rtyuio}
&&\frac{1}{4\pi^{5/2}}\sum_{N,K=-\infty}^{\infty}\exp\left(iNx+iKy\right)\exp\left[-\frac{1}{2}\left(N^2+K^2\right)\right]
\int_{-\infty}^{\infty}dz_I\exp\left(-z_I^2+\frac{Nz_I}{2}-\frac{Kz_I}{ 2}\right)\nonumber\\
&&\times\int_{0}^{2\pi}dz_R\exp\left(-\frac{iKz_R}{2}-\frac{iNz_R}{2}\right)\nonumber\\
&&=\frac{1}{2\pi^{3/2}}\sum_{N,K=-\infty}^{\infty}\exp(iNx+iKy)
\int_{-\infty}^{\infty}dz_I\exp\left(-z_I^2+\frac{Nz_I}{2}-\frac{Kz_I}{ 2}\right)\exp\left[-\frac{1}{2}\left(N^2+K^2\right)\right]\delta(K,-N)
\nonumber\\&&=\delta(x-y)
\end{eqnarray}
\end{proof}
Using this proposition we find that the scalar product is given by
\begin{eqnarray}\label{A3}
\bra{q_2}q_1\rangle=\frac{1}{2\pi}\int _Adm(z)Q_1(z)[Q_2(z)]^*
\end{eqnarray}
Also
\begin{eqnarray}\label{pa25}
q(x)=\int _Adm(z)Q(z)\Theta_3\left [\frac{x-z^*}{2};\frac{i}{2\pi} \right] 
\end{eqnarray}

\section{The reproducing kernel formalism for finite quantum systems}
Given a `fiducial state' $\ket{f}$, let $F(z)$ be its analytic representation.
Below we consider the $d^2$ states ${\cal D}(\alpha,\beta)\ket {f}$ in the analytic representation.
The fiducial vector should not be a position or a momentum state because in this case 
many of the ${\cal D}(\alpha,\beta)\ket {f}$ differ only by a phase factor, and represent the same physical state.
The fiducial state should be a `generic vector'.
Using an expansion of the fiducial state analogous to Eq.(\ref{6}), we find that the overlap of two coherent states is 
\begin{eqnarray}
\bra{f}{\cal D}(-\gamma,-\delta){\cal D}(\alpha,\beta)\ket{f}=
\omega \left [\frac{1}{2}(\alpha \beta +\gamma \delta)-\beta \gamma\right ]
\sum _n f_{n+\beta -\delta }^*f_n \omega [(\alpha -\gamma )n].
\end{eqnarray}

The analytic functions representing the states ${\cal D}(\alpha,\beta)\ket {f}$, are
\begin{eqnarray}\label{aaa}
&&{\mathfrak D}(z; \alpha, \beta; f)=\pi^{-1/4} \sum_{m=0}^{d-1} \bra{X;m}{\cal D}(\alpha,\beta)\ket {f}\;\Theta_3 \left [\frac{\pi m}{d}-z\sqrt{\frac{\pi}{2d}};
\frac{i}{d}\right ];\;\;\;\;\alpha, \beta \in {\mathbb Z}(d).
\end{eqnarray}
The $f$ in the notation, indicates the dependence on the fiducial state.
They obey periodicity relations analogous to Eq.(\ref{periodicity}):
\begin{eqnarray}
 &&{\mathfrak D}(z+\sqrt{2\pi d}; \alpha, \beta; f)= {\mathfrak D}(z; \alpha, \beta; f)\nonumber\\
 &&{\mathfrak D}(z+i\sqrt{2\pi d}; \alpha, \beta;f)= {\mathfrak D}(z; \alpha, \beta; f)\exp \left (\pi d-i z\sqrt{2\pi d} \right ).
\end{eqnarray}
For a fixed fiducial vector $\ket{f}$, the set of the $d^2$ analytic functions ${\mathfrak D}(z; \alpha, \beta; f)$ with $\alpha, \beta\in {\mathbb Z}(d)$,
are the analogue in the present context, of the coherent states for a harmonic oscillator.

\begin{proposition}\label{pro1}
\mbox{}
\begin{itemize}
\item[(1)]
The ${\mathfrak D}(z; \gamma, \delta; f)$ is a two-dimensional Fourier transform of ${\mathfrak D}(-z; \alpha, \beta; f)$: 
\begin{eqnarray}\label{13c}
{\mathfrak D}(z; \gamma, \delta; f)=\frac{1}{d}\sum _{\alpha, \beta}\omega (-2^{-1}\beta \gamma +2^{-1}\alpha \delta){\mathfrak D}(-z; \alpha, \beta; f)
\end{eqnarray}
\item[(2)]
The ${\mathfrak D}(z; \alpha, \beta; f)$ of 
the coherent state ${\mathcal D}(\alpha, \beta)\ket {f}$ is related to $F(z)$ of the fiducial vector as follows:
\begin{eqnarray}\label{13bb}
{\mathfrak D}(z; \alpha, \beta; f)=\omega (-2^{-1}\alpha \beta)F\left (z-\beta\sqrt{\frac{2\pi}{d}} +i \alpha\sqrt{\frac {2\pi }{d}}\right )\exp \left (iz\alpha \sqrt{\frac{2\pi}{d}}
-\frac{\pi\alpha ^2}{d}\right)
\end{eqnarray}
The zeros $\zeta _{\nu}$ of the analytic representation $F(z)$ of the fiducial state, are related to the zeros $\zeta _{\nu}(\alpha, \beta)$ of ${\mathfrak D}(z; \alpha, \beta; f)$, as follows:
\begin{eqnarray}\label{14a}
\zeta _{\nu}(\alpha , \beta )=\zeta _{\nu}-i\alpha\sqrt{\frac{2\pi }{d}}+\beta\sqrt{\frac{2\pi }{d}}
;\;\;\;\;\alpha, \beta \in {\mathbb Z}(d)
\end{eqnarray}
\item[(3)]
The resolution of the identity in the language of analytic representation is 
\begin{eqnarray}\label{21}
\frac{1}{d}\sum _{\alpha, \beta}{\mathfrak D}(z; \alpha, \beta; f)[{\mathfrak D}(w; \alpha, \beta; f)]^*=K(z,w^*)
\end{eqnarray}
where $K(z,w^*)$ is the reproducing kernel, and is given by
\begin{eqnarray}\label{22}
&&K(z,w^*)=\pi ^{-1/2}\sum_{m=0}^{d-1} \Theta_3 \left (\frac{\pi m}{d}-z\sqrt{\frac{\pi}{2d}};\frac{i}{d}\right )
\Theta_3 \left (\frac{\pi m}{d}-w^*\sqrt{\frac{\pi}{2d}};\frac{i}{d}\right )\nonumber\\
&&K(z,w^*)=K(w^*,z);\;\;\;\;\;\;K(z,w^*)=K(-z,w^*)
\end{eqnarray}
$K(z,w^*)$ does not depend on the fiducial vector $\ket{f}$.
\item[(4)]Reproducing kernel relation:
For any $G(z)$ in the space ${\cal A}$
\begin{eqnarray}\label{23}
G(z)=\frac{1}{d^{3/2}\sqrt{2\pi } }\int_S d\mu (w) K(z,w^*)G(w).
\end{eqnarray}
\item[(5)]
$G(z)$ can be written in terms of the expansions
\begin{eqnarray}\label{24a}
G(z)=\frac{1}{d}\sum _{\alpha, \beta}{\mathfrak D}(z; \alpha, \beta; f)g(\alpha, \beta; f);\;\;\;\;\;g(\alpha, \beta; f)=\bra{f}{\cal D}(-\alpha, -\beta) \ket{g}
\end{eqnarray}
and also 
\begin{eqnarray}\label{24b}
G(z)=\frac{1}{d}\sum _{\gamma , \delta}{\mathfrak D}(-z; \gamma , \delta; f){\widetilde g}(\gamma , \delta; f);\;\;\;\;\;\;
{\widetilde g}(\gamma , \delta; f)&=&\bra{f}{\cal P}(-2^{-1}\gamma, -2^{-1}\delta ) \ket{g}
\end{eqnarray}
The inverses of these relations are:
\begin{eqnarray}\label{29a}
g(\alpha, \beta; f)=\frac{1}{d^{3/2}\sqrt{2\pi }}\int_S d\mu (w)[{\mathfrak D}(w; \alpha, \beta; f)]^*G(w)
\end{eqnarray}
and
\begin{eqnarray}\label{29b}
{\widetilde g}(\gamma , \delta; f)=\frac{1}{d^{3/2}\sqrt{2\pi }}\int_S d\mu (w)[{\mathfrak D}(-w;\gamma, \delta;f)]^*G(w)
\end{eqnarray}
The ${\widetilde g}(\gamma , \delta; f)$ is related to
$g(\alpha, \beta; f)$ through a two-dimensional Fourier transform
\begin{eqnarray}\label{38}
{\widetilde g}(\gamma , \delta; f)=\frac{1}{d}\sum _{\alpha, \beta}
g(\alpha, \beta; f)\omega (2^{-1}\beta \gamma -2^{-1}\alpha \delta)
\end{eqnarray}
\item[(6)]
The following equations, which can be called `marginal properties', relate our analytic representation to the $X$- and $P$-representation:
\begin{eqnarray}\label{marg}
&&\frac{1}{d}\sum _{\alpha =0}^{d-1}{\mathfrak D}(z; \alpha, 2\beta; f)=\pi^{-1/4} f_{-\beta}\;\Theta_3 \left [\frac{\pi \beta}{d}-z\sqrt{\frac{\pi}{2d}};
\frac{i}{d}\right ]\nonumber\\
&&\frac{1}{d}\sum _{\beta =0}^{d-1}{\mathfrak D}(z; 2\alpha, \beta; f)=
\pi^{-1/4}{\widetilde f}_{-\alpha}\exp \left (-\frac{z^2}{2}\right )\Theta_3 \left [\frac{\pi \alpha}{d}-iz\sqrt{\frac{\pi}{2d}};\frac{i}{d}\right ].
\end{eqnarray}

\end{itemize}
\end{proposition}

\begin{proof}
\mbox{}
\begin{itemize}
\item[(1)]
We first define the
\begin{eqnarray}\label{aaa10}
&&{\mathfrak P}(z; \alpha, \beta ;f)=\pi^{-1/4} \sum_{m=0}^{d-1} \bra{X;m}{\cal P}(\alpha,\beta)\ket {f}\;\Theta_3 \left [\frac{\pi m}{d}-z\sqrt{\frac{\pi}{2d}};\frac{i}{d}\right ]
\end{eqnarray}
Using Eq.(\ref{45}), we prove that
\begin{eqnarray}\label{13a}
{\mathfrak P}(z; \gamma, \delta ;f)=\frac{1}{d}\sum _{\alpha , \beta}\omega (\beta \gamma -\alpha \delta){\mathfrak D}(z; \alpha, \beta;f).
\end{eqnarray}

For odd $d$, we use the relations
\begin{eqnarray}
&&\bra{X;m}{\cal D}(\alpha,\beta)=\omega (-2^{-1}\alpha \beta+\alpha m)\bra {X;m-\beta}\nonumber\\
&&\bra{X;m}{\cal P}(\alpha,\beta)
=\omega (-2\alpha \beta+2\alpha m)\bra {X;-m+2\beta}
\end{eqnarray}
to prove that
\begin{eqnarray}\label{13b}
{\mathfrak P}(z; \alpha, \beta;f)={\mathfrak D}(-z; -2\alpha, -2\beta;f).
\end{eqnarray}
We then insert Eq.(\ref{13b}) into Eq.(\ref{13a}) and we get Eq.(\ref{13c}).
\item[(2)]
We first point out that
\begin{eqnarray}
 \bra{X;m}{\cal D}(\alpha,\beta)\ket {f}=\omega (-2^{-1}\alpha \beta +\alpha m)f_{m-\beta}
\end{eqnarray}
Therefore
\begin{eqnarray}\label{33a}
{\mathfrak D}(z; \alpha, \beta;f)&=&\pi^{-1/4} \sum_{m=0}^{d-1}\omega (-2^{-1}\alpha \beta +\alpha m)f_{m-\beta} \;\Theta_3 \left [\frac{\pi m}{d}-z\sqrt{\frac{\pi}{2d}};\frac{i}{d}\right ]\nonumber\\
&=&\pi^{-1/4} \omega (2^{-1}\alpha \beta )\sum_{m=0}^{d-1}f_m\omega (\alpha m)\;\Theta_3 \left [\frac{\pi (m+\beta)}{d}-z\sqrt{\frac{\pi}{2d}};\frac{i}{d}\right ]
\end{eqnarray}
We then show that
\begin{eqnarray}
&&\omega (\alpha m)\;\Theta_3 \left [\frac{\pi (m+\beta)}{d}-z\sqrt{\frac{\pi}{2d}};\frac{i}{d}\right ]\nonumber\\&&=
\omega (-\alpha \beta)\exp\left [i\alpha z\sqrt{\frac{2\pi }{d}}-\frac{\pi\alpha ^2}{d}\right ]
\Theta_3 \left [\frac{\pi m}{d}-z\sqrt{\frac{\pi}{2d}}+\frac{\beta \pi }{d}-i\frac{\alpha \pi }{d};\frac{i}{d}\right ]
\end{eqnarray}
From this follows Eq.(\ref{13bb}).
Eq.(\ref{14a}) follows immediatelly from Eq.(\ref{13bb}).

\item[(3)]
The proof is based on the property \cite{R2}
\begin{eqnarray}\label{689}
\frac{1}{d}\sum _{\alpha,\beta}{\cal D}(\alpha,\beta)\ket {f}\bra{f}[{\cal D}(\alpha,\beta)]^{\dagger}={\bf 1}
\end{eqnarray}
From this follows that
\begin{eqnarray}
&&\frac{1}{d}\Theta_3 \left (\frac{\pi m}{d}-z\sqrt{\frac{\pi}{2d}};\frac{i}{d}\right )\left [\sum _{\alpha,\beta}{\cal D}(\alpha,\beta)\ket {f}
\bra{f}[{\cal D}(\alpha,\beta)]^{\dagger}\right ]\Theta_3 \left (\frac{\pi n}{d}-w^*\sqrt{\frac{\pi}{2d}};\frac{i}{d}\right )
 \nonumber\\&&={\bf 1} \Theta_3 \left (\frac{\pi m}{d}-z\sqrt{\frac{\pi}{2d}};\frac{i}{d}\right )
\Theta_3 \left (\frac{\pi n}{d}-w^*\sqrt{\frac{\pi}{2d}};\frac{i}{d}\right ).
\end{eqnarray}
Therefore
\begin{eqnarray}
&&\Theta_3 \left (\frac{\pi m}{d}-z\sqrt{\frac{\pi}{2d}};\frac{i}{d}\right )\bra{X;m}\left [\sum _{\alpha,\beta}{\cal D}(\alpha,\beta)\ket {f}
\bra{f}[{\cal D}(\alpha,\beta)]^{\dagger}\right ]\ket{X;n}\Theta_3 \left (\frac{\pi n}{d}-w^*\sqrt{\frac{\pi}{2d}};\frac{i}{d}\right )
 \nonumber\\&&=d\delta (m,n)\Theta_3 \left (\frac{\pi m}{d}-z\sqrt{\frac{\pi}{2d}};\frac{i}{d}\right )
\Theta_3 \left (\frac{\pi n}{d}-w^*\sqrt{\frac{\pi}{2d}};\frac{i}{d}\right ).
\end{eqnarray}
From this follows that
\begin{eqnarray}
&&\sum _{\alpha,\beta}\Theta_3 \left (\frac{\pi m}{d}-z\sqrt{\frac{\pi}{2d}};\frac{i}{d}\right )\bra{X;m}{\cal D}(\alpha,\beta)\ket {f}
\bra{f}[{\cal D}(\alpha,\beta)]^{\dagger}\ket{X;n}\Theta_3 \left (\frac{\pi n}{d}-w^*\sqrt{\frac{\pi}{2d}};\frac{i}{d}\right )
 \nonumber\\&&=d\delta (m,n)\Theta_3 \left (\frac{\pi m}{d}-z\sqrt{\frac{\pi}{2d}};\frac{i}{d}\right )
\Theta_3 \left (\frac{\pi n}{d}-w^*\sqrt{\frac{\pi}{2d}};\frac{i}{d}\right ).
\end{eqnarray}
Summation over $m$ and $n$ gives in part of Eq.(\ref{21}).

\item[(4)]
In order to prove Eq.(\ref{23}) we insert Eq.(\ref{22}) into Eq.(\ref{23}) and use Eq.(\ref{5b}).
\item[(5)]
Eqs(\ref{24a}), (\ref{24b}) are proved using the resolution of the identity.
In order to prove Eq.(\ref{29a}) we insert it into Eq.(\ref{24a}) and use Eq.(\ref{21}).
In a similar way we prove Eq.(\ref{29b}).
For the prrof of Eq.(\ref{38}) we insert it into Eq.(\ref{24b}) and we use Eq.(\ref{13c}).
\item[(6)]
Using Eq.(\ref{33a}) we get
\begin{eqnarray}
\frac{1}{d}\sum _{\alpha}{\mathfrak D}(z; \alpha,2 \beta;f)
&=&\pi^{-1/4} \sum _{m}f_m\;\Theta_3 \left [\frac{\pi (m+2\beta)}{d}-z\sqrt{\frac{\pi}{2d}};\frac{i}{d}\right ]
\frac{1}{d}\sum_{\alpha=0}^{d-1}\omega (\alpha m+\alpha \beta)\nonumber\\
&=&\pi^{-1/4} \sum _{m}f_m\;\Theta_3 \left [\frac{\pi (m+2\beta)}{d}-z\sqrt{\frac{\pi}{2d}};\frac{i}{d}\right ]\delta (m+\beta,0)\nonumber\\
&=&\pi^{-1/4} f_{-\beta}\;\Theta_3 \left [\frac{\pi \beta}{d}-z\sqrt{\frac{\pi}{2d}};\frac{i}{d}\right ]
\end{eqnarray}
In order to prove the second equation we use again Eq.(\ref{33a}) and we get
\begin{eqnarray}
\frac{1}{d}\sum _{\beta}{\mathfrak D}(z; 2\alpha, \beta;f)&=&\pi^{-1/4}\frac{1}{d}\sum _{m,\beta} \omega (\alpha \beta +2\alpha m) \;
f_{m}\;\Theta_3 \left [\frac{\pi (m+\beta)}{d}-z\sqrt{\frac{\pi}{2d}};\frac{i}{d}\right ]\nonumber\\&=&
\pi^{-1/4}\frac{1}{d^{3/2}}\sum _{m,\beta, \kappa} \omega (\alpha \beta +2\alpha m+m\kappa) \;
{\widetilde f}_{\kappa}\;\Theta_3 \left [\frac{\pi (m+\beta)}{d}-z\sqrt{\frac{\pi}{2d}};\frac{i}{d}\right ]
\end{eqnarray}
We change the variables $(m, \beta)$ into $(\mu=m+\beta, \lambda=m-\beta)$.
This is bijective map from ${\mathbb Z}(d)\times {\mathbb Z}(d)$ into ${\mathbb Z}(d)\times {\mathbb Z}(d)$.
The existence of $2^{-1}$ in ${\mathbb Z}(d)$ with odd $d$, is important in proving this.
Therefore we get
\begin{eqnarray}
\frac{1}{d}\sum _{\beta}{\mathfrak D}(z; 2\alpha, \beta;f)&=&
\pi^{-1/4}\frac{1}{d^{3/2}}\sum _{\kappa, \mu, \lambda} \omega (2^{-1}\lambda \alpha +2^{-1}\lambda \kappa) \; \omega (2^{-1}3\alpha \mu+2^{-1}\mu \kappa) \;
{\widetilde f}_{\kappa}\;\Theta_3 \left [\frac{\pi \mu}{d}-z\sqrt{\frac{\pi}{2d}};\frac{i}{d}\right ]\nonumber\\&=&
\pi^{-1/4}\frac{1}{d^{1/2}}\sum _{\kappa, \mu} \delta (\alpha +\kappa)\; \omega (2^{-1}3\alpha \mu+2^{-1}\mu \kappa) \;
{\widetilde f}_{\kappa}\;\Theta_3 \left [\frac{\pi \mu}{d}-z\sqrt{\frac{\pi}{2d}};\frac{i}{d}\right ]\nonumber\\&=&
\pi^{-1/4}\frac{1}{d^{1/2}}{\widetilde f}_{-\alpha}\sum_{\mu}\omega (\alpha \mu)\;\Theta_3 \left [\frac{\pi \mu}{d}-z\sqrt{\frac{\pi}{2d}};\frac{i}{d}\right ]
\nonumber\\&=&
\pi^{-1/4}{\widetilde f}_{-\alpha}\exp \left (-\frac{z^2}{2}\right )\Theta_3 \left [\frac{\pi \alpha}{d}-iz\sqrt{\frac{\pi}{2d}};\frac{i}{d}\right ]
\end{eqnarray}
Eq.(\ref{aaa1}) has been used in the last step.
\end{itemize}
\end{proof}

\subsection{Physical meaning of proposition \ref{pro1}}\label{591}

The physical meaning of the various parts of the above proposition is as follows:
\begin{itemize}
\item[(1)]
${\mathfrak P}(z; \alpha, \beta ;f)$ in Eq.(\ref{aaa10}) are coherent states with fiducial vector ${\cal P}(0,0)\ket{f}$.
Unlike the Glauber coherent states where ${\cal P}(0,0)\ket{0}=\ket{0}$, here ${\cal P}(0,0)\ket{f}$ is in general different from $\ket{f}$. Consequently,
in Eq.(\ref{13a}) the coherent states ${\mathfrak P}(z; \alpha, \beta ;f)$ are related to the coherent states ${\mathfrak D}(z; \alpha, \beta;f)$
(with fiducial vector $\ket{f}$) through a two-dimensional Fourier transform.
Eq.(\ref{13b}) shows that the coherent states ${\mathfrak P}(z; \alpha, \beta ;f)$ are also the coherent states ${\mathfrak D}(-z; -2\alpha, -2\beta;f)$.
All these statements are distilled into Eq.(\ref{13c}), which shows that there is a two-dimensional Fourier transform between 
${\mathfrak D}(z; \gamma, \delta; f)$ and ${\mathfrak D}(-z; \alpha, \beta; f)$.
There is no analogue of this equation for standard coherent states \cite{KL}.
\item[(2)]
Eq.(\ref{13bb}) shows that the zeros of 
${\mathfrak D}(z; \alpha, \beta; f)$ are the zeros of 
the fiducial state $F(z)$
displaced by $(\frac{L\alpha}{d}, \frac{L\beta}{d})$, where $L$ is the length of each side of the cell $S$.
So the zeros of all ${\mathfrak D}(z; \alpha, \beta; f)$ form a $d\times d$ lattice within the torus.
\item[(3)]
This is the analogue of the resolution of the identity, in the language of the analytic representation.
It is important that the $K(z,w^*)$ does not depend on the fiducial vector $\ket{f}$. 
We have explained earlier, that the
fiducial vector should be a `generic vector' (i.e., it should not be a position or a momentum state).
\item[(4)]
Eq.(\ref{23}) shows that the $K(z,w^*)$ is indeed the reproducing kernel.
\item[(5)]
Eqs.(\ref{24a}),(\ref{24b}) analyze an arbitrary state in terms of coherent states, and the coefficients are given in Eqs.(\ref{29a}),(\ref{29b}).
\item[(6)]
Eq.(\ref{marg}) are the marginal properties of the displacement operators. 

\end{itemize}

\subsection{Coherent states with ${\cal F}\ket{f}$ as fiducial vector}

In general, if we change the fiducial vector we get a different set of coherent states.
Here we consider the fiducial vector ${\cal F}\ket{f}$, and show that we get the same set of coherent states.

In analogy to Eq.(\ref{aaa}), we introduce the $d^2$ analytic functions ${\mathfrak F}(z; \alpha, \beta;f)$, which are defined a s follows:
\begin{eqnarray}
{\mathfrak F}(z; \alpha, \beta;f)=\pi^{-1/4} \sum_{m=0}^{d-1} \bra{X;m}{\cal F}(\alpha,\beta)\ket {f}\;\Theta_3 \left [\frac{\pi m}{d}-z\sqrt{\frac{\pi}{2d}};
\frac{i}{d}\right ]
\end{eqnarray}
Eq.(\ref{4}) shows that they are coherent states with ${\cal F}\ket{f}$ as fiducial vector.
They are related to ${\mathfrak D}(z; \alpha, \beta; f)$, as follows:
\begin{eqnarray}\label{123}
{\mathfrak F}\left [z; -\frac{1}{2}(\alpha-\beta), -\frac{1}{2}(\alpha+\beta); f\right ]=\omega \left [\frac{1}{4}(\alpha^2+ \beta^2)\right]\exp \left (\frac{-z^2}{2}\right )
{\mathfrak D}(iz; \alpha, \beta; f)
\end{eqnarray}
Indeed 
\begin{eqnarray}
&&{\mathfrak F}(z; \alpha, \beta;f)=\pi^{-1/4} \sum_{m} \bra{X;m}{\cal F}(\alpha,\beta)\ket {f}\;\Theta_3 \left [\frac{\pi m}{d}-z\sqrt{\frac{\pi}{2d}};
\frac{i}{d}\right ]\nonumber\\
&&=\pi^{-1/4} \frac{1}{\sqrt{d}}\omega \left [\frac{1}{2}(\alpha ^2+\beta ^2)\right ]\sum_{m,n} \omega (mn)\bra{X;n}{\cal D}(-\alpha-\beta, \alpha-\beta)
\ket {f}\;\Theta_3 \left [\frac{\pi m}{d}-z\sqrt{\frac{\pi}{2d}};\frac{i}{d}\right ]
\end{eqnarray}
We then use Eq.(\ref{aaa1}) and we prove Eq.(\ref{123}).

\section{The reproducing kernel formalism for systems on a circle}

Coherent states on a circle have been studied in \cite{co1,co2,co3}.
Here we approach them using the language of analytic representations.
Let $\ket{r}$ be a `fiducial state' with analytic representation $R(z)$.
The fiducial state should be a `generic vector' (not a position or a momentum state).
Coherent states are defined as 
\begin{eqnarray}
\ket {a, K}_{\rm coh}=D(a, K)\ket{r};\;\;\;\;\;\ket {a +2\pi, K}_{\rm coh}=(-1)^K\ket {a, K}_{\rm coh}
\end{eqnarray}
The overlap of two coherent states is 
\begin{eqnarray}
_{\rm coh}\langle b, M\ket {a, K}_{\rm coh}&=&\frac{1}{2\pi}\int _0 ^{2\pi}dx \;r(x)\;r^*(x+a -b)\nonumber\\&\times&
\exp\left[i(K-M)x+i\left (\frac{Ka}{2}+\frac{Mb}{2}-Ma\right )\right ]
\end{eqnarray}
\begin{proposition}
The coherent states obey the resolution of the identity
\begin{eqnarray}\label{pa12}
\frac{1}{2\pi}\sum_{K=-\infty}^{\infty}\int_0^{2\pi} da \;\ket {a, K}_{\rm coh\;\;coh}\bra{a, K}={\bf 1}.
\end{eqnarray}
\end{proposition}
\begin{proof}
We calculate the matrix elements of the operators in Eq.(\ref{pa12}) with momentum states, and we get
\begin{eqnarray}
&&\sum_{K=-\infty}^{\infty}\int_0^{2\pi} da \langle M|D(a,K)|r\rangle\langle r|D(-a,-K)|N\rangle\nonumber\\&&
=\sum_{K=-\infty}^{\infty}\int_0^{2\pi}da\exp \left [-ia\left (M-\frac{K}{2}\right )\right ]
\exp\left [ia\left (N-\frac{K}{2}\right )\right ]r_{M-K}r^*_{N-K}\nonumber\\&&=
\sum_{K=-\infty}^{\infty}\int_0^{2\pi}da \exp[ia(N-M)]r_{M-K}r^*_{N-K}=2\pi \delta_{NM}\sum_{K=-\infty}^{\infty}|r_{M-K}|^2=2\pi \delta_{NM}
\end{eqnarray}
\end{proof}
\begin{remark}
In the special case of fiducial vectors such that 
\begin{eqnarray}
\sum _{K=-\infty}^\infty |r_{\sigma +\tau K}|^2=\frac{1}{\tau};\;\;\;\;\;\sigma=0,..., \tau -1
\end{eqnarray}
we get the following $\tau$ resolutions of the identity
\begin{eqnarray}\label{pa12}
\frac{\tau}{2\pi}\sum_{K=-\infty}^{\infty}\int_0^{2\pi} da \;\ket {a, \tau K+\sigma}_{\rm coh\;\;coh}\bra{a, \tau K +\sigma}={\bf 1};\;\;\;\;\;\sigma=0,..., \tau -1.
\end{eqnarray}
The proof is analogous to the above.
\end{remark}
The coherent states $\ket {a, K}_{\rm coh}$ are represented by the analytic functions
\begin{eqnarray}\label{aaa}
&&{\mathfrak d}(z; a, K; r)=\int _0^{2\pi}dx\bra{x}{D}(a,K)\ket {r}\;
\Theta_3 \left [\frac{x-z}{2};\frac{i}{2\pi}\right ];\;\;\;\;a\in {\mathbb S};\;\;\;\;\; K\in {\mathbb Z}\nonumber\\
&&{\mathfrak d}(z+2\pi; a, K; r)={\mathfrak d}(z; a, K; r)
\end{eqnarray}
The $r$ in the notation, indicates the dependence on the fiducial state.
Here
\begin{eqnarray}\label{rrr}
{\mathfrak d}(z; a +2\pi, K; r)=(-1)^K{\mathfrak d}(z; a, K; r)
\end{eqnarray}

\begin{proposition}\label{pro10}
\begin{itemize}
\item[(1)]
The analytic representation $\mathfrak{d}(z;a,K;r)$ of 
the coherent state $\ket {a, K}_{\rm coh}$ can be written as a two-dimensional Fourier transform of $\mathfrak{d}(z;b,M;r)$: 
\begin{eqnarray}\label{pa1}
\mathfrak{d}(-z;a,K;r)=\frac{1}{2\pi}\sum_{M=-\infty}^\infty \int_0^{2\pi}db\;\mathfrak{d}(z;b,2M-K;r)\exp\left[\frac{i}{2}(-bK-aK+2Ma)\right]
\end{eqnarray}
\item[(2)]
The $\mathfrak{d}(z;a,K;r)$ of 
the coherent state $\ket {a, K}_{\rm coh}$ is related to $R(z)$ of the fiducial vector as follows:
\begin{eqnarray}\label{pa3}
\mathfrak{d}(z;a,K;r)=\exp\left(-\frac{1}{2}iKa+iKz-\frac{1}{2}K^2\right)R\left(z+iK-a\right).
\end{eqnarray}
The zeros $\zeta _n$ of $R(z)$ are related to the zeros $\zeta_n(a,K)$ of $\mathfrak{d}(z;a,K;r)$, as follows:
\begin{eqnarray}\label{zerosfgt}
\zeta_n(a,K)=\zeta_n-iK+a
\end{eqnarray}
\item[(3)]
The resolution of the identity is
\begin{eqnarray}\label{pa20}
&&\frac{1}{4\pi ^2}\sum_{K=-\infty}^{\infty}\int_0^{2\pi} da\;\mathfrak{d}(z;a,K;r)[\mathfrak{d}(w;a,K;r)]^*=K_c(z,w^*)
\end{eqnarray}
where
\begin{eqnarray}\label{Kernel}
&&K_c(z,w^*)=\int_0^{2\pi}dx\;\Theta_3\left[\frac{x-z}{2};\frac{i}{2\pi} \right]\Theta_3\left[\frac{x-w^*}{2};\frac{i}{2\pi} \right]\nonumber\\
&&K_c(z,w^*)=K_c(-z,-w^*)
\end{eqnarray}
is the reproducing kernel. The index $c$ indicates `circle'.
\item[(4)]
The reproducing kernel relation is given by
\begin{eqnarray}\label{rwe}
Q(z)=\int_{A}dm(w)\;K_c(z,w^*)Q(w).
\end{eqnarray}
\item[(5)]
$Q(z)$ can be written as
\begin{eqnarray}\label{z1}
&&Q(z)=\frac{1}{2\pi}\sum_{K=-\infty}^{\infty}\int_{0}^{2\pi}da\;\mathfrak{d}(z;a,K;r)q(a,K;r)
;\;\;\;\;q(a,K;r)=\langle r|D(-a,-K)|q\rangle\nonumber\\
&&q(a+2\pi,K;r)=(-1)^Kq(a,K;r)
\end{eqnarray}
and also
\begin{eqnarray}\label{z2}
&&Q(z)=\frac{1}{2\pi}\sum_{M=-\infty }^\infty \int_{0}^{2\pi}db\;\mathfrak{d}(-z;b,M;r)\widetilde q(b,M;r);\;\;\;\;\;
\widetilde{q}(b,M;r)=\langle r|U(-b,-M)|q\rangle\nonumber\\
&&\widetilde{q}(b+2\pi,M;r)=(-1)^M\widetilde{q}(b,M;r)
\end{eqnarray}
The inverse of these relations are
\begin{eqnarray}\label{z3}
q(a,K;r)=\frac{1}{2\pi}\int_{A}\;dm(w)\mathfrak{d}(w;a,K;r)^*Q(w),
\end{eqnarray}
and 
\begin{eqnarray}\label{z4}
\widetilde{q}(b,M;r)=\frac{1}{2\pi}\int_{A}dm(w)\;[\mathfrak{d}(-w;b,M;r)]^*Q(w).
\end{eqnarray}
The $\widetilde{q}(b,M;r)$ is related to $q(a,K;r)$ as follows
\begin{eqnarray}\label{z5}
\widetilde{q}(b,M;r)=\frac{1}{2\pi}\sum_{K=-\infty}^{\infty}\int_{0}^{2\pi}da\; q(-a,M-2K;r)\exp\left[\frac{i}{2}(-aM-bM+2Kb)\right].
\end{eqnarray}

\item[(6)]
The relation between the analytic representation and $X$ and $P$ representations is given by the following `marginal properties':
\begin{eqnarray}\label{790}
&&\sum_{K=-\infty}^{\infty}\mathfrak{d}(z;a,K;r)=2\pi\; r\left(-\frac{1}{2}a\right)\Theta_3\left[\frac{a-2z}{4};\frac{i}{2\pi} \right]\nonumber\\
&&\int_0^{2\pi}da\mathfrak{d}(z;a,-2K;r)=4\pi^2 r_K\exp\left[-izK-\frac{1}{2}K^2\right]
\end{eqnarray}

\end{itemize}
\end{proposition}
\begin{proof}
\mbox{}
\begin{itemize}
\item[(1)]
Using Eqs.(\ref{aaa}), (\ref{pa2}) we get
\begin{eqnarray}
\mathfrak{d}(-z;-a,-K;r)=\;\int_{0}^{2\pi}dx\langle x|U(a,K)|r\rangle\Theta_3\left[\frac{x-z}{2};\frac{i}{2\pi} \right]
\end{eqnarray}
Inserting Eq.(\ref{fou}) we get
\begin{eqnarray}
\mathfrak{d}(-z;-a,-K;r)&=&\frac{1}{2\pi}\int_{0}^{2\pi}dx\sum_{M=-\infty}^\infty \int_0^{2\pi}db\;\langle x|D(b,K+2M)|r\rangle \nonumber\\&\times& \Theta_3\left[\frac{x-z}{2};\frac{i}{2\pi} \right]
\exp\left[\frac{i}{2}(bK-aK-2aM)\right].
\end{eqnarray}
from which follows Eq.(\ref{pa1}).
\item[(2)]
We use Eqs.(\ref{pa2}),(\ref{aaa}) and we get
\begin{eqnarray}\label{1290}
\mathfrak{d}(z;a,K;r)&=&\int_{0}^{2\pi}dx\;r(x)\exp\left(iKx+\frac{1}{2}iKa\right)\nonumber\\
&&\hspace{3cm}\times\Theta_3\left[\frac{x+a-z}{2};\frac{i}{2\pi}\right ].
\end{eqnarray}
We next use the definition of Theta function in Eq. (\ref{pa4}) and we get
\begin{eqnarray}\label{560}
\mathfrak{d}(z;a,K;r)&=&\exp\left(\frac{1}{2}iKa\right)\sum_{N=-\infty}^{\infty}\exp\left(iNa-iNz-\frac{1}{2}N^2\right)\nonumber\\
&\times&\int_{0}^{2\pi}dx\;r(x)\exp[(iK+iN)x]
\end{eqnarray}
We then change variables $N+K=M$ and eventually we get
\begin{eqnarray}\label{thet123}
\mathfrak{d}(z;a,K;r)&=&\exp\left(-\frac{1}{2}iKa+iKz-\frac{1}{2}K^2\right)\nonumber\\
&&\times\int_{0}^{2\pi}dx\;r(x)\Theta_3\left[\frac{x}{2}-\frac{1}{2}(z+iK-a);\frac{i}{2\pi}\right]
\end{eqnarray}
from which follows Eq.(\ref{pa3}).

\item[(3)]
Inserting Eq.(\ref{aaa}) into the left hand side of Eq.(\ref{pa20}) we get 
\begin{eqnarray}
&&\sum_{K=-\infty}^{\infty}\int_0^{2\pi} da\int_0^{2\pi}dx \langle x|D(a,K)|r\rangle\Theta_3\left[\frac{x-z}{2};\frac{i}{2\pi}\right ]
\int_0^{2\pi}dy\langle r|D(-a,-K)|y\rangle\Theta_3\left[\frac{y-w^*}{2};\frac{i}{2\pi} \right]\nonumber\\
&&=4\pi^2\int_0^{2\pi}dx \int_0^{2\pi}dy \;\delta(x,y)\Theta_3\left[\frac{x-z}{2};\frac{i}{2\pi}\right ]\Theta_3\left[\frac{y-w^*}{2};\frac{i}{2\pi} \right]
\end{eqnarray}
The resolution of the identity of Eq.(\ref{pa12}) has been used in the proof of this equality.
From this follows Eq.(\ref{pa20}).

\item[(4)]
We insert Eq.(\ref{Kernel}) into Eq.(\ref{rwe}) and using Eq.(\ref{pa25}) we get
\begin{eqnarray}
Q(z)&=&\int_A dm(w)Q(w)\int_0^{2\pi}dx\;\Theta_3\left[\frac{x-z}{2};\frac{i}{2\pi}\right ]\Theta_3\left[\frac{x-w^*}{2};\frac{i}{2\pi} \right]\nonumber\\
&=&\int_0 ^{2\pi} dxq(x)\Theta_3\left[\frac{x-z}{2};\frac{i}{2\pi}\right ].
\end{eqnarray}
This proves Eq.(\ref{rwe}).

\item[(5)]
The proof of Eqs.(\ref{z1}), (\ref{z2}) is based on the resolution of the identity property of coherent states.
In order to prove Eq.(\ref{z3}), we insert it into Eq.(\ref{z1}) and we use the resolution of the identity in the form of Eq.(\ref{pa20}).
In the same way we prove Eq.(\ref{z4}).

In order to prove Eq.(\ref{z5}), we insert it into Eq.(\ref{z2}) and we use Eqs.(\ref{pa1}), (\ref{pa20}).

\item[(6)]
We insert Eq.(\ref{1290}) into the first of Eqs.(\ref{790}) and we get
\begin{eqnarray}
\sum_{K=-\infty}^{\infty}\mathfrak{d}(z;a,K;r)&=&\int_{0}^{2\pi}dxr(x)\Theta_3\left[\frac{x+a-z}{2};\frac{i}{2\pi} \right]\sum_{K=-\infty}^{\infty}\exp\left(iKx+\frac{1}{2}iKa\right)\nonumber\\
&&=2\pi\; r\left(-\frac{1}{2}a\right)\Theta_3\left[\frac{2^{-1}a-z}{2};\frac{i}{2\pi} \right]
\end{eqnarray}
The second of Eqs.(\ref{790}) is proved in an analogous way.

\end{itemize}
We have checked that all equations in this proposition are consistent with Eq.(\ref{rrr}).

\end{proof}

The physical meaning of proposition \ref{pro10} is analogous to the one discussed in section \ref{591}.
But there are differences which we now point out:
\begin{itemize}

\item
Eq.(\ref{zerosfgt}) shows that the zeros of all ${\mathfrak d}(z; a,K;r)$ are on horizontal lines at integer values of the imaginary axis.
\item
For the marginal properties, we can compare and contrast Eqs.(\ref{marg}),(\ref{790}).
\end{itemize}

\section{Wigner and Weyl functions}
The whole phase space formalism \cite{ZACHOS} can be expressed in the analytic language.
In this section we express briefly the Wigner and Weyl functions in terms of the coefficients $g(\alpha, \beta ;f)$ that describe the state $\ket {g}$ in the analytic language for the finite systems,
and also in terms of the coefficients $q(a,K;r)$ that describe the state $\ket {q}$ in the analytic language for the systems on circles.
\subsection{Finite systems}
The Weyl and Wigner functions of a state $\ket{g}$ are defined as
\begin{eqnarray}
&&{\widetilde W}(g;\alpha, \beta )=\bra{g}{\cal D}(\alpha,\beta)\ket{g};\;\;\;\;\;W(g;\alpha, \beta )=\bra {g}{\cal P}(\alpha,\beta)\ket{g}\nonumber\\
&&W(g;\gamma, \delta)=\frac{1}{d} \sum _{\alpha, \beta}\omega (\beta \gamma -\alpha \delta){\widetilde W}(g;\alpha,\beta)
\end{eqnarray}
\begin{proposition}\label{pro3}
The Weyl function is given in terms of the $g(\alpha, \beta ;f)$ by
\begin{eqnarray}\label{W1}
{\widetilde W}(g;\alpha, \beta )=\frac{1}{d}\sum _{\gamma, \delta}g(\gamma, \delta;f)g^*(-\alpha +\gamma, -\beta +\delta;f)\omega[2^{-1}(\alpha \delta-\beta \gamma)]
\end{eqnarray}
The Wigner function is given in terms of the $g(\alpha, \beta ;f)$ by
\begin{eqnarray}\label{W2}
{W}(g;\alpha, \beta )=\frac{1}{d^2}\sum _{\gamma, \delta, \epsilon, \zeta}g(\epsilon, \zeta ;f)g^*(\gamma, \delta;f)
\omega(\alpha \delta-\beta \gamma +2^{-1} \zeta \gamma -\zeta \alpha -2^{-1}\epsilon \delta +\epsilon \beta)
\end{eqnarray}
\end{proposition}
\begin{proof}
Using Eq.(\ref{689}) we get
\begin{eqnarray}
&&{\widetilde W}(g;\alpha, \beta )=\bra{g}{\cal D}(\alpha,\beta)\ket{g}=\frac{1}{d}\sum _{\gamma, \delta}\bra{g}{\cal D}(\gamma , \delta)\ket{f}
\bra{f}{\cal D}(-\gamma , -\delta){\cal D}(\alpha,\beta)\ket{g}
\end{eqnarray}
We then use Eq.(\ref{24a}) to prove Eq.(\ref{W1}).

For the Wigner function, we use Eq.(\ref{689}) and we get
\begin{eqnarray}
{W}(g;\alpha, \beta )&=&\bra{g}{\cal P}(\alpha,\beta)\ket{g}=\frac{1}{d}\sum _{\gamma, \delta}\bra{g}{\cal D}(\gamma , \delta)\ket{f}
\bra{f}{\cal D}(-\gamma , -\delta){\cal P}(\alpha,\beta)\ket{g}\nonumber\\
&=&\frac{1}{d}\sum _{\gamma, \delta}\bra{g}{\cal D}(\gamma , \delta)\ket{f}
\bra{f}{\cal P}(\alpha-2^{-1}\gamma ,\beta -2^{-1}\delta)\ket{g}\omega (\alpha \delta -\beta \gamma)
\end{eqnarray}
We then use Eqs.(\ref{24a}),(\ref{24b}) to prove Eq.(\ref{W2}).
\end{proof}

\subsection{Systems on a circle}

The Wigner and Weyl functions of a state $|q\rangle$ on a circle is defined as
\begin{eqnarray}\label{Wigner10}
W(a,K;q)&=&\langle q|U(a,K)|q\rangle;\;\;\;\;\;\widetilde{W}(a,K;q)=\langle q|D(a,K)|q\rangle\nonumber\\
W(a,M;q)&=&\frac{1}{2\pi}\sum_{K=-\infty}^{\infty}\int_0^{2\pi}db\widetilde{W}(b,M+2K;q)\exp\left[\frac{i}{2}(bM-aM-2Ka)\right]\nonumber\\
W(a+2\pi,K;q)&=&(-1)^KW(a,K;q);\;\;\;\;\;\widetilde{W}(a+2\pi,K;q)=(-1)^K\widetilde{W}(a,K;q)
\end{eqnarray}
\begin{proposition}
The Weyl function  is given in terms of the $q(a,K;r)$ 
\begin{eqnarray}\label{KOKK}
\widetilde{W}(a,K;q)&=&\frac{1}{2\pi}\sum_{M=-\infty}^{\infty}\int_0^{2\pi} db\; q^*(b,M;r)q(-a+b,-K+M;r)\nonumber\\
&&\hspace{2cm}\times\exp\left[\frac{i}{2}\left(-aM+Kb\right)\right]
\end{eqnarray}

The Wigner function is given in terms of the $q(a,K;r)$ 
\begin{eqnarray}\label{23456}
&&W(a,K;q)=\frac{1}{4\pi^2}\sum_{M,N=-\infty}^{\infty}\int_0^{2\pi} db\int_0^{2\pi} d\gamma\;[q(b,M;r)]^*q(-\gamma,-K+M-2N;r)\nonumber\\
&&\hspace{3cm}\times\exp\left[\frac{i}{2}(\gamma K-\gamma M-aK-2aN+2bK-bM+2bN)\right]
\end{eqnarray}

\end{proposition}
\begin{proof}
Using Eq.(\ref{pa12}) we get
\begin{eqnarray}
\widetilde{W}(a,K;q)=\langle q|D(a,K)|q\rangle=\frac{1}{2\pi}\sum_{M=-\infty}^{\infty}\int_0^{2\pi} db\; \langle q|D(b,M)|r\rangle\langle r|D(-b,-M)D(a,K)|q\rangle
\end{eqnarray}
We then use Eqs.(\ref{pa2}), (\ref{z1}) to prove Eq.(\ref{KOKK})

For the Wigner function we use Eq.(\ref{pa12}) and we get
\begin{eqnarray}\label{kokok}
W(a,K;q)&=&\langle q|U(a,K)|q\rangle=\frac{1}{2\pi}\sum_{M=-\infty}^{\infty}\int_0^{2\pi} db\; \langle q|D(b,M)|r\rangle\langle r|D(-b,-M)D(a,K)U_0|q\rangle\nonumber\\
\end{eqnarray}
We then use Eqs.(\ref{pa2}), (\ref{z2}), (\ref{z5}) to prove Eq.(\ref{23456})

\end{proof}

\section{Discussion}

We have considered quantum systems with $d$-dimensional Hilbert space, where $d$ is an odd integer.
The formalism uses the $2^{-1}$ which exists in ${\mathbb Z}(d)$ with odd $d$.
We have represented the states of such systems with the analytic functions in Eq.(\ref{aaa1}) which obeys the boundary conditions of Eq.(\ref{periodicity}), 
and therefore it is effectively defined on a torus. The scalar product is given in Eq.(\ref{scalar}).

We have also discussed analogous formalism for systems on a circle.
For simplicity we have used periodic boundary conditions (zero Aharonov-Bohm magnetic flux).
Here the states are represented with the analytic functions on a strip in Eq.(\ref{A1}) which obeys the boundary conditions of Eq.(\ref{A2}). 
The scalar product is given in Eq.(\ref{A3}).

We have studied the reproducing kernel formalism for these two systems, in proposition \ref{pro1} for the finite case, and in proposition \ref{pro10} for the circle.
These two propositions are the main results of this paper.
We have also studied the Wigner and Weyl functions, in this language.

There are other applications of theta functions in various topics in quantum physics.
They include the Heisenberg-Weyl groups \cite{A}, discrete Fourier transforms\cite{R},
quantum theta functions\cite{M},  applications to quantum field theory\cite{T}, etc.
There are also applications in harmonic analysis and time-frequency analysis\cite{H1,H2,H3}, algebraic number theory \cite{FR},
automorphic forms \cite{BU}, etc.

In this paper we used theta functions in analytic representations of quantum systems on ${\mathbb Z}(n)$ and on a circle.
The results can be used for further studies of these systems.

\paragraph*{Acknowledgement:}Helpful discussions with Professor R.F. Bishop are gratefully acknowledged


\begin{thebibliography}{90}

\bibitem{B}
V. Bargmann, Commun. Pure Appl. Math. {\bf 14}, 187 (1961)
\bibitem{P}
A. Perelomov, `Generalized coherent states and their applications', (Springer, Berlin, 1986) 
\bibitem{B1}
S. Bergman, `The kernel function and conformal mapping' (Amer. Math. Soc, Rhode Island, 1970)
\bibitem{H}
B.C. Hall, Contemp. Math. 260, 1 (2000)
\bibitem{V}
A. Vourdas, J. Phys. A39, R65 (2006) 
\bibitem{ZV}
S. Zhang, A. Vourdas, J. Phys. A37, 8349 (2004); and corrigendum in J. Phys. A38, 1197 (2005) 
\bibitem{TVZ}
M. Tubani, A. Vourdas, S. Zhang, Phys. Scr. 82, 038107 (2010)
\bibitem{CGV}
N. Cotfas, J-P Gazeau, A. Vourdas, J. Phys. A44, 175303 (2011)
\bibitem{T1}
D Mumford, `Tata lectures on Theta', Vols 1,2,3 (Birkhauser, Boston, 1983)
\bibitem{T2}
J. Igusa, `Theta functions' (Berlin, Springer, 1972)
\bibitem{T3}
J. Fay, `Theta functions on Riemann surfaces' (Springer, Berlin, 1973) 
\bibitem{L}
P. Leboeuf, A. Voros, J. Phys. A23, 1765 (1990)
\bibitem{R1}
A. Vourdas, Rep. Prog. Phys. 67, 1 (2004)
\bibitem{R2}
A. Vourdas, J. Phys. A40, R285 (2007)
\bibitem{R3}
M. Kibler, J. Phys. A42, 353001 (2009)
\bibitem{R4}
N. Cotfas, J.P. Gazeau, J.Phys. A43, 193001(2010)
\bibitem{R5}
T. Durt, B.G. Englert, I. Bengtsson, K. Zyczkowski, Int. J. Quantum Comp. 8, 535 (2010)
\bibitem{R6}
P. Stovicek, J. Tolar, Rep. Math. Phys. 20, 157 (1984) 
\bibitem{R7}
J. Tolar, G. Hadzitaskos, J. Phys. A30, 2509 (1997)
\bibitem{c1}
M.G.G. Laidlaw, C Morette-De Witt, Phys. Rev. D{\bf 3}, 1375 (1971)
\bibitem{c2}
J.S. Dowker J. Phys. A{\bf 5}, 936 (1972)
\bibitem{c3}
L.S. Schulman, J. Math.Phys. {\bf 12}, 304 (1971) 
\bibitem{c4}
L.S. Schulman, "Techniques and applications of path integration" (Wiley,  New York, 1981).
\bibitem{c5}
F. Acerbi, G. Morchio and F. Strocchi, Lett. Math. Phys. {\bf 27}, 1 (1993)
\bibitem{c6}
F. Acerbi, G. Morchio and F. Strocchi, J. Math. Phys. {\bf 34}, 889 (1993)
\bibitem{c7}
H. Narnhofer, W.Thirring,  Lett. Math. Phys. {\bf 27}, 133(1993)
\bibitem{c8}
S. Zhang, A. Vourdas, J. Math. Phys. 44, 5084 (2003)
\bibitem{co1}
K.Kowalski, J.Rembielinski, L.C.Papaloucas, J. Phys. A29, 4149 (1996)
\bibitem{co2}
J.A. Gonzalez, M.A. del Olmo, J. Phys. A31, 8841 (1998)
\bibitem{co3}
K.Kowalski, J.Rembielinski, J. Phys. A35, 1405 (2002)
\bibitem{co4}
G. Chatzitaskos, P. Luft, J. Tolar, J. Phys: Conf. Series, 284, 012016 (2011)
\bibitem{KL}
J. Klauder, B-S Skagerstam, `Coherent states' (World scientific, Singapore, 1985)
\bibitem{ZACHOS}
C.K. Zachos, D.B. Fairlie, T.L. Curtright, `Quantum Mechanics in Phase Space' (World Scientific, Singapore, 2005)
\bibitem{A}
R. Tolimery, Trans. Am. Math. Soc. 239, 293 (1978)
\bibitem{R}
M. Ruzzi, J. Math. Phys. 47, 063507 (2006)
\bibitem{M}
F. Luef, Y. Manin, Lett. Math. Phys. 88, 131 (2009)
\bibitem{T}
A. Tyurin, `Quantization, classical and quantum field theory and theta functions' (American Math. Society, Rhode Island, 2003)
\bibitem{H1}
K. Gr\"ochening `Foundations of time-frequency analysis' (Birkh\"auser, Boston, 2001)
\bibitem{H2}
W. Schempp, Proc. Am. Math. Soc. 92, 103 (1984)
\bibitem{H3}
L. Auslander, R. Tolimieri, SIAM J. Math. Anal. 16, 577 (1985)
\bibitem{FR}
S. Friedberg, J. Number Theo. 20,121 (1985)
\bibitem{BU}
D. Bump, `Automorphic forms and representations' (Cambridge Univ. Press, Cambridge 1998)
\end{thebibliography}
\end{document}